\documentclass[a4paper]{article}       %

\usepackage{lineno}
\usepackage[english]{babel}
\usepackage[linesnumbered, longend ]{algorithm2e}
\usepackage{graphicx}
\usepackage{amssymb}
\usepackage{amsmath}
\usepackage{amsthm}
\usepackage{enumitem}

\usepackage{epstopdf}
\usepackage{pgf, tikz}

\ifpdf
\else 
\fi

\usepackage{hyperref}
\usepackage{theoremref}
\usepackage{float}
\usepackage[utf8]{inputenc}
\usepackage[numbers, square, comma]{natbib}
\renewcommand{\S}{\mathbb S}
\newcommand{\R}{\mathbb R}
\newcommand{\eps}{\varepsilon}

\newcommand{\argmin}{\operatornamewithlimits{argmin}}

\DeclareMathOperator{\diam}{diam}

\DeclareMathOperator{\Span}{span}

\DeclareMathOperator{\nbot}{\not\hspace{-1mm}\bot}

\newcommand{\lb}{\textsc{Lb}}
\newcommand{\Mid}{\textsc{Mid}}

\newcommand{\LO}{\mathcal{L}}

\theoremstyle{plain}
\newtheorem{lemma}{Lemma}
\newtheorem{theorem}[lemma]{Theorem}

\theoremstyle{definition}
\newtheorem{remark}[lemma]{Remark}

\usepackage{color}

\begin{document}


\title{An algorithm for computing Fréchet means on the sphere}

\author{Gabriele Eichfelder\footnote{Institute for Mathematics, Technische Universität Ilmenau, Weimarer Straße 25, 98693~Ilmenau, Germany, \{gabriele.eichfelder, thomas.hotz\}@tu-ilmenau.de}	\and
        Thomas Hotz$^*$		
        \and
        Johannes Wieditz\footnote{Institute for Mathematical Stochastics, Georg-August-Universität Göttingen, Goldschmidtstraße 7, 37077 Göttingen, Germany, johannes.wieditz@uni-goettingen.de}
}

\date{\today}

\maketitle


\begin{abstract}
\noindent
For most optimisation methods an essential assumption is the vector space structure of the feasible set. This condition is not fulfilled if we consider optimisation problems over the sphere. We present an algorithm for solving a special global problem over the sphere, namely the determination of Fréchet means, which are points minimising the mean distance to a given set of points. The Branch and Bound method derived needs no further assumptions on the input data, but is able to cope with this objective function which is neither convex nor differentiable. The algorithm's performance is tested on simulated and real data.\\

\end{abstract}

\section{Introduction}
\label{sec:1}

Data assuming values on the unit sphere $\S^2$, also known as \textit{spherical data}, arise frequently in applications, for example as directions of remanent magnetisation in soil samples, as arrival directions of cosmic particle showers, or as wind directions. 
For general information on the statistical analysis of spherical data, see e.g. \cite{Fisher87}, \cite{Mardia2000}.

When considering optimisation problems on the sphere with these input data, one quickly encounters problems because most approaches fail due to the missing vector space structure. Neither the sum of two points of the sphere, nor the multiplication with a real number is an element of the sphere anymore. 

A well-known example for such an optimisation problem arises from non-Euclidean statistics. There, one aim is to define a \textit{mean} on the sphere. Because of the reasons named above, how to define such a \textit{mean} is at first sight unclear. Thus, often a more intrinsic point of view is used, namely considering the sphere as a Riemannian manifold with arc length as the metric measuring distances between two points. A mean may then be defined to be a point which minimises the average distance to some power measured to a set of given data points~$x_1,\dots,x_n\in\S^2$, i.e. an element of
\begin{equation}
\label{eq:1}
\argmin\bigg\{ \frac1n \sum_{i=1}^n d^p(m,x_i) \, \bigg\vert \, m\in \S^2  \bigg\}
\end{equation}
for a fixed non-negative number $p$. This is the set of so-called \textit{Fréchet-$p$-means}. 
Note that on a Euclidean space, the case $p=2$ leads to the usual notion of mean, whereas the case $p=1$ results in a (spatial) median.

While computing Fr\'echet-$2$-means on the 1-sphere, i.e. the circle, is well investigated and efficiently possible, cf. \cite{Hotz2015}, \cite{McKilliam12}, the case of the 2-sphere turns out to be quite challenging. Even for real numbers $p>1$, the objective function is neither differentiable nor convex in the sense of manifolds \cite{Afsari11}. 
Therefore, Fr\'echet-$p$-means need not be unique; however, since the objective function is continuous and the sphere is compact, there does always exist a Fr\'echet-$p$-mean.

In the literature, there are many different ways to tackle the problem of computing a Fréchet mean on the sphere. A probabilistic method is introduced in \cite{Arnaudon2016}, where the authors use a Monte Carlo approach. Moreover, \cite{Afsari13} developed a steepest descent method for the case that the given points are located in one half of the sphere which implies that the unique Fr\'echet-2-mean lies in that half sphere where the objective function in addition is convex. 

Furthermore, there exists a variety of problems similar to (\ref{eq:1}) each making slightly different assumptions or focussing on different aspects. If $d$ is the Euclidean distance, there are several results about lower bounds and characterising means of the function $m\mapsto \sum_{i=1}^n d^p(x_i,m)$ where the latter needs some further assumptions for the points, cf. \cite{Stolarsky75}, \cite{Wagner90}, \cite{Weiszfeld2009}. Fréchet functions and Fréchet means over more general spaces were considered in \cite{Fletcher2008} and \cite{Palfia2011}. An essential assumption there is an upper bound for the so-called local injectivity radius which is violated in our case of the 2-sphere.



A deterministic approach for computing the set of Fréchet means on the 2-sphere for general configurations of points based on~(\ref{eq:1}) has not been considered so far;
general deterministic optimisation algorithms on manifolds appear to be local search algortithms, see e.g. \cite{Absil2007}.

To fill this gap, we will in the following introduce a Branch and Bound method which is quite universal in the sense that it is able to cope with the non-differentiability and non-convexity of the objective function, requiring no assumptions on the given points $x_1,\dots,x_n\in\S^2$. Our approach will be outlined in Section~\ref{sec:2} while lower bounds for the objective function will be established in Section~\ref{sec:3}. Using these bounds, we can then introduce the algorithm in Section~\ref{sec:4}; there we also consider numerical and efficiency aspects. We illustrate the performance of the algorithm in Section~\ref{sec:5} using simulations as well as an application to real data, discussing the results obtained in Section~\ref{sec:6}.

\section{Fréchet means on the sphere and a branching scheme}
\label{sec:2}

The optimisation problem we study here is that of finding so-called Fréchet-$p$-means on the unit sphere. Therefore we consider the sphere as a metric space $(\S^2,d)$, endowed with arc length as the distance, a finite number of points~$x_1,\dots,x_n\in \S^2$ and a real number $p\ge 0$. Then we are looking for the set of all minimisers of the program
\begin{equation}
\label{eq:2}
\min_{m\in \S^2} \quad \frac1n \sum_{i=1}^n d^p(m,x_i)
\end{equation}
where $d(x,y) = \arccos\langle x, y\rangle$ for $x, y \in \S^2$, $\langle\cdot,\cdot\rangle$ denoting the standard dot product in $\R^3$.
The objective function $\hat F_n: \S^2 \to \R, m \mapsto \frac 1n\sum_{i=1}^n d^p (m, x_i)$ of~(\ref{eq:2}) is also referred to as \emph{Fréchet-$p$-function} (or Fr\'echet function in short). 

One can consider (\ref{eq:2}) more generally on metric spaces. If these are fulfilling the Heine-Borel property (i.e. every closed bounded subset is compact), such minimisers always exist, although uniqueness is not guaranteed in general. Hereafter, we will focus on the spherical case in which at least existence is ensured.

The main idea of Branch and Bound in continuous optimisation is to divide (\textit{branch}) the feasible set step-by-step into smaller subsets. One then tries to eliminate, under usage of suitable lower bounds, subsets which cannot contain a minimum (\textit{bound}). Using this method we obtain an approximation of the set of all minimisers; here, we will use the definition of an \textit{$(\eps,\delta)$-approximation}. Recall that for $\eps,\delta>0$ a set $A\subseteq \S^2$ is an $(\eps,\delta)$-approximation of the set of minimisers $X$ of a function $f: \S^2\to\R$ if for all $a\in A$ it holds that $f(a) - \min_{m\in \S^2} f(m) \le \eps$ and for all $x\in X$ there exists an $a\in A$ with $d(a,x)< \delta$.

To apply a Branch and Bound algorithm to our problem, we have to first specify a rule how to subdivide the sphere as well as a discarding rule. For our purposes it is most appropriate to divide the sphere in spherical triangles. These are generated by a triple of non-coplanar vectors which are called \textit{vertices} of the triangle. 
More precisely, a spherical triangle is the intersection of the (closed) convex cone spanned by its vertices with the sphere.
Here, we start with the triangles induced by the vertices of a regular octahedron inscribed in the sphere as an initial triangulation. A triangle is then divided in one branch step at its midpoint of the longest side; in case of non-uniqueness we choose a side according to a deterministic rule. The so generated sequence of triangles fulfils the criterion of exhaustiveness, i.e. their diameters converge to 0, which is needed to show convergence of the algorithm, cf. \cite[p. 204 ff.]{Locatelli2013}. Furthermore, it is easily possible to determine whether a given point lies within a triangle using the following elementary geometric observation.

\begin{lemma}
\label{lemma:1}
	Let $x\in\S^2$ and $\Delta\subseteq\S^2$ be a spherical triangle with non-coplanar vertices $d_1,d_2,d_3\in\S^2$. 
	\begin{enumerate}[label=(\alph*)]
		\item Then $x\in \Delta$ if and only if the solution $\lambda = (\lambda_1, \lambda_2, \lambda_3) \in \R^3$ of the linear equations $\sum_{i=1}^3 \lambda_i d_i = x$ is component-wise non-negative.
		\item Moreover, let $n_{ij}\in\S^2$ be orthogonal to $\Span\{d_i,d_j\}$, $i \neq j$, with the property that for the third vector $d_k$, $k\neq i,j$ we have $\langle d_k, n_{ij} \rangle \ge 0$. Then $x\in\Delta$ if and only if $\langle x, n_{ij}\rangle \ge 0$ for all $i,j\in\{1,2,3\}$, $i \neq j$.
	\end{enumerate}
\end{lemma}

\begin{proof}
(a) is the algebraic formulation of $x$ being in $\Delta$ if and only if $x$ is in the convex cone spanned by $d_1, d_2, d_3$.
In fact, the latter cone is the intersection of the three (closed) half-spaces determined by the conditions $\langle x, n_{ij} \rangle \ge 0$ in~(b).
\end{proof}


\section{Lower bounds for the Fréchet function}
\label{sec:3}
The second ingredient necessary for designing a Branch and Bound method is an appropriate discarding rule. For this, we first note the following simple result. 

\begin{lemma}
\label{lemma:2}
	Let $g$ be a lower bound of our objective function $f: \S^2\to \R$  on a subset $\Delta\subsetneq \S^2$ and let $\bar m \in \S^2\setminus \Delta$. If $f(\bar m) < g$ holds, then $\Delta$ cannot contain a minimiser of $f$.	
\end{lemma}
Hence, our aim is to construct a lower bound for our objective function, the Fréchet function, on a spherical triangle. 

Since the Fréchet function is Lipschitz continuous the construction of a lower bound using the Lipschitz constant is possible. However, this approach leads to lower bounds which are too weak for our purposes. We thus will now construct sharper estimates.

Having non-negative lower bounds $g_i$ on $\Delta$ for the functions $m\mapsto d(m,x_i)$ for each $i\in\{1,\dots,n\}$ we obtain a lower bound for the Fréchet function via 
\begin{equation}
\label{eq:3}
	\frac1n \sum_{i=1}^n g_i^p.
\end{equation}
In fact, we can calculate the minimum $\min_{m\in \Delta} d(m,x_i)$ of the distance function on a spherical triangle analytically. This is of course the best lower bound for the distance function. The following proposition shows how to calculate this minimum.
\begin{theorem}
	
	\label{prop:1}
	
	Let $x\in\S^2$ and consider a spherical triangle $\Delta\subseteq\S^2$ with non-coplanar vertices $d_1,d_2,d_3\in\S^2$. Furthermore, for $i,j\in\{1,2,3\}$, $i \neq j$ let
	\begin{equation}
	\label{eq:4}
	g_{ij} = \begin{cases}
        \big|\arcsin\big( \langle x, d_i\times d_j \rangle \big)\big|, &\text{ if } x = \lambda_1 d_i + \lambda_2 d_j + \lambda_3 \left( d_i\times d_j \right)\\[-1pt]& \text{ for some }\lambda_1,\lambda_2\ge 0, \lambda_3 \in \R, \\[3pt]
	\min\big\{ d(x,d_i), d(x,d_j) \},&\text{ otherwise }
	\end{cases}
	\end{equation}
	which is the distance of $x$ to the great circle arc connecting $d_i$ and $d_j$.
	
	Then for the distance of $x$ to $\Delta$ we have
	\begin{equation}
	\label{eq:5}
	d(x,\Delta) = \begin{cases}
	0, & \text{if }x \in \Delta,\\
	\min\left\{g_{12},g_{23}, g_{31}\right\}, & \text{otherwise}.
	\end{cases}
	\end{equation} 	
	In particular, the average $\frac1n \sum_{i=1}^n d(x_i,\Delta)$ is a lower bound for $\hat F_n$ on $\Delta$.
\end{theorem}

\begin{proof}
	
	Equation~(\ref{eq:5}) holds for sure in case of $x\in\Delta$. Otherwise, because of monotonicity arguments the minimum will be attained at the boundary of $\Delta$ which are three great circle arcs. Therefore, it is sufficient to show that for the distance of $x$ to such an arc Equation~(\ref{eq:4}) holds. We consider w.l.o.g. the great circle arc $\kappa$ between $d_1$ and $d_2$. 
	
	If $x\bot \Span\{d_1,d_2\}$, then $x=\pm \frac1{\|d_1\times  d_2\|} \left( d_1\times d_2 \right)$ and every arc connecting $x$ and $\kappa$ has the same length $\frac \pi 2$. Thus, $d(x,\kappa) = \frac \pi 2 = g_{12}$.
	
	In the case $x \nbot \Span\{d_1,d_2\}$, we compute at first the distance of $x$ to the whole great circle $\Gamma$ (containing $\kappa$), see Figure~\ref{fig:half}. The point $\bar m$ minimising the distance in arc length from $x$ to $\Gamma$ is the same as the one minimising the Euclidean distance from $x$ to $\Gamma$ since both result from each other by a monotonous transformation. For the latter one, application of Pythagoras' theorem leads to $\|x - \bar m\|^2 = \|x - p\|^2 + \| p - \bar m\|^2$ where $p$ is the orthogonal projection of $x$ onto $\Span\{d_1,d_2\}$. Because the first term is constant, only the second term has to be minimised and the point $\bar m\in \Gamma$ with the smallest Euclidean distance to $p$ is $\bar m = \frac1{\|p\|} p$. This exists and is unique since $x \nbot \Span\{d_1,d_2\}$. The distance from $x$ to $\Gamma$ is therefore given by $d(x,\Gamma)=d(x,\bar m)$. 
	
	\begin{figure}
		\centering
		\includegraphics[width=.7\textwidth]{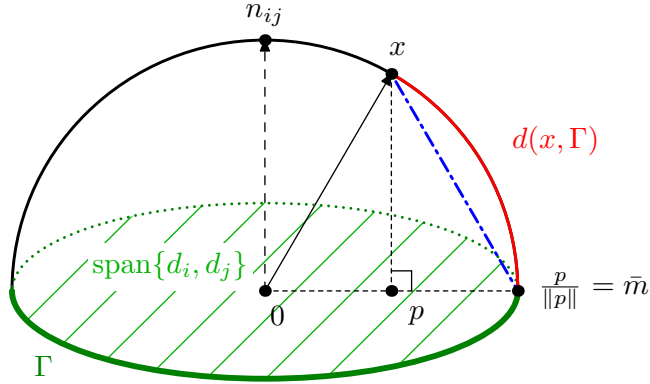}
		\caption{ Upper half of the sphere with marked distance (red) of $x\in\S^2$ from the great circle $\Gamma$ induced by $\Span\{d_i,d_j\}$ (green). The minimiser $\bar m$ is given by $\argmin\left\{ \|x-m\|: m\in \Gamma \right\} = \frac p{\|p\|}$. The Euclidean distance $\|x - \bar m\|$ is marked in blue ($- \cdot $)}
		\label{fig:half}
	\end{figure}
	
	Assume now, that $n_{12} = \frac{d_1 \times d_2}{\Vert d_1 \times d_2 \Vert}$ (which -- possibly up to sign -- equals the corresponding vector in Lemma~\ref{lemma:1}) and $x$ lie in the same half space induced by $\Span\{d_1,d_2\}$. Then $n_{12}$, $x$ and $\bar m$ lie on the same great circle arc and we have $d(\bar m, x)+d(x,n_{12}) = d(\bar m, n_{12}) = \frac \pi 2$. Using the identity $\arccos(\cdot) +\arcsin(\cdot) = \frac \pi 2$ we obtain $d(x,\Gamma) = \arcsin \langle x,d_1\times d_2\rangle$. The case $-\left(d_1\times d_2\right)$ and $x$ lying in the same half space can be treated analogously using the property $\arcsin(-t)=-\arcsin(t)$ for all $t\in[-1,1]$. 
	
	We consider now the computation of the distance of $x$ to the great circle arc $\kappa$. If the minimiser $\bar m = \argmin_{m\in \Gamma}d(x,m) \in \kappa$, then we can compute the distance as shown above. Otherwise it can be shown, under usage of some monotonicity arguments, that one of the endpoints $d_1,d_2$ of the arc has the smallest distance to $x$. 
	
	It remains to show that $\bar m \in \kappa$ is equivalent to the first condition in Equation~(\ref{eq:4}) holds. For that, we consider the great circle arc $\kappa$ connecting $d_1,d_2$ and assume $x\nbot\Span\{d_1,d_2\}$. From our reasoning above we already know that the minimiser of $\min_{m\in\Gamma}d(x,m)$ is obtained via normalising the orthogonal projection of $x$ to $\Span\{d_1,d_2\}$. The points obtained from normalising points of $\kappa$ can be written as $\lambda_1 d_1 + \lambda_2 d_2$ with $\lambda_1,\lambda_2\ge 0$. Since the orthogonal part to $\Span\{d_1,d_2\}$ vanishes after projection, any point $x$ having a minimiser in $\kappa$ can be written as $\lambda_1 d_1 + \lambda_2 d_2 + \lambda_3 \left( d_1 \times d_2 \right)$ with $\lambda_1,\lambda_2\ge 0$ and $\lambda_3\in\R$ which is the first condition in Equation~(\ref{eq:4}). 
\end{proof}
\pagebreak
\begin{remark}
One can easily show that this lower bound fulfils the exactness in the limit property, i.e. if the diameter of the triangle converges to zero then the bound converges to the function value at the limit point, cf. \cite{Locatelli2013}. Thus the convergence of our method follows with the aid of elementary estimates and the triangle inequality. Furthermore, it is efficiently possible to check the conditions in Equation~(\ref{eq:4}) and (\ref{eq:5}) using the results of Lemma~\ref{lemma:1}.
\end{remark}

\section{The Spherical Branch and Bound algorithm ($\S$BB algorithm)}
\label{sec:4}

The algorithm derived here is based on ideas from \cite{Eichfelder2016}, where a similar algorithm was presented which aimed at finding globally optimal minimisers, but over a box in a linear space. Adapting this algorithm to our optimisation problem over the sphere leads to the following \textit{Spherical Branch and Bound algorithm} or short \textit{$\S$BB algorithm} presented below. In the following, we explain the basic steps in detail.

\begin{algorithm}
	\SetKwData{Left}{left}\SetKwData{This}{this}\SetKwData{Up}{up}
	\SetKwFunction{Union}{Union}\SetKwFunction{FindCompress}{FindCompress}
	\SetKwInOut{Input}{Input}\SetKwInOut{Output}{Output}
	
	\Input{Data points $x_1,x_2,\dots, x_n$, $p\ge 0$, accuracies $\eps>0,\delta>0$, initial subdivision $\Delta_1, \Delta_2,\dots,\Delta_k$ of the sphere $(\S^2,d)$ where $\bigcup_{i=1}^k \Delta_i \supseteq \S^2$}
	\Output{$(\eps,\delta)$-approximation of the set of Fréchet-$p$-means}
	\BlankLine
	
	Initialise list $\LO = \left\{ X_1 = \{ \Delta_1,s_1,v_1, u_1 \}, \dots, X_k = \{ \Delta_k, s_k, v_k, u_k  \}  \right\} $ where $s_i=\diam(\Delta_i)$ is the diameter of $\Delta_i$, $v_i = \hat F_n(\Mid(\Delta_i))$ and $u_i = \lb(\Delta_i)$. \\
	Set $A = \emptyset$, $X^\star = \{\Delta_1, s_1, v_1, u_1\} $, $X_{\mathrm{act}} = X^\star$, $u^\star = -\infty$, $v_{\mathrm{act}} = \infty$, $v_{\mathrm{glob}}=v_{\mathrm{act}}$. \\
	
	\While{$\LO \neq \emptyset$}{
		
		$\LO \leftarrow \LO \setminus \left\{ X^\star \right\} $\\
		$sB(X^\star) \leftarrow \textsc{Branch}( X^\star )$
		
		\For{$\mathrm{all }$ $\bar X = \{ \bar \Delta,\bar s, \bar v, \bar u \} \in sB(X^\star)$ }{
			
			Compute value $\bar v$ of $\hat F_n$ in $\Mid(\bar X)$, diameter $\bar s$ of $\bar X$ and lower bound $\bar u = \lb(\bar X)$.  \\
			\If{ $\bar u \le v_{\mathrm{glob}}$}{
				
				Add $ \{ \bar \Delta, \bar s, \bar v, \bar u \}$ at the end of $\LO$.\\
				\If{ $\bar v \le v_{\mathrm{act}}$ }{
					$X_{\mathrm{act}} \leftarrow \bar X$, 
					$v_{\mathrm{act}} \leftarrow \bar v$,
					$v_{\mathrm{glob}} \leftarrow \min\{ v_{\mathrm{act}}, v_{\mathrm{glob}} \}$\\
					Delete all elements $\{\Delta,s,v,u\}$ from $\LO$ where $u>v_{\mathrm{glob}}$.
					
				}
				
			}
		}
		\If{$\LO \neq \emptyset$ }{
			
			Define $X^\star = \{\Delta^\star, s^\star, v^\star,  u^\star\}$ as the first element of $\LO$ with $u^\star = \min_{ \{\Delta,s, v, u\}\in \LO } u$.\\
			\While{$\LO \neq \emptyset \wedge v_{act} - u^\star \le \frac\eps2$}{
				\If{$X_{\mathrm{act}}\in\LO \wedge s_{\mathrm{act}} \le \delta$ }{
					Set $A \leftarrow A \cup \left\{ X_{\mathrm{act}}\right\}$.\\
					Delete $X_{\mathrm{act}}$ from $\LO$.\\
				}
				
				\If{$X_{\mathrm{act}}\in\LO$ } { 
					$X^\star \leftarrow X_{\mathrm{act}}$, go to line~3. \\ 
				}

				\If{$\LO \neq \emptyset$ }{
					Define $X^\star\leftarrow \{ \Delta^\star, s^\star, v^\star, u^\star \}$ as the first element of $\LO$ with $u^\star = \min_{\{\Delta,s,v,u\}\in \LO } u$.\\
					Compute $X_{\mathrm{act}}$ as the first element in $\argmin_{ \{\Delta,s, v, u\}\in \LO } v$ and update upper bound $v_{\mathrm{act}} \leftarrow \hat F_n(\Mid(X_{\mathrm{act}}))$.\\
				}
			}
		}	
		
	}
	
	\BlankLine
	\vfill
	\caption{$\S$BB algorithm for computing a $(\eps,\delta)$-approximation of all Fréchet-$p$-means}
	\label{alg:algorithm}
\end{algorithm}

Initially, we start with generating a list of all elements to be visited. Every item consists of a description of the triangle $\Delta$, its diameter $s$ as well as a lower and upper bound for the objective function on this set. Here, the upper bound is calculated as the function's value at a certain point $\textsc{Mid}(\Delta)$. We use the centroid of a spherical triangle for this. For the lower bound~$\textsc{Lb}$, we use the bound obtained from Theorem~\ref{prop:1}.

The outer loop (line 3--31) contains
\begin{itemize}
	\item a branch part (line 5), in which we subdivide the currently chosen triangle using the \textsc{Branch} rule of Section~\ref{sec:2},
	
	\item a bounding part (line 6--15), where elements of the branch part will be added to the list~$\LO$ and elements fulfilling the discarding criterion from Lemma~\ref{lemma:2} with the lower bound \textsc{Lb} obtained from (\ref{eq:3}) combined with~(\ref{eq:5}), will be eliminated from $\LO$,
		
	\item a selection part (line 16--32), where the next element of the list is selected. Then we decide whether this element can be added to the list $A$ of the $(\eps,\delta)$-approximation (line 18--21) or if it has to be revisited in a later iteration. In the first case, we delete the element from $\LO$ and select the element with the smallest lower bound in $\LO$ as the new current element (line 25--28).
\end{itemize} 

For an efficient implementation it is important to accelerate the computation steps as much as possible. An important key role plays the calculation of the lower bounds on the triangles resulting from the branch steps. 

To compute the lower bound, we have to make several tests. For a triangle~$\Delta$ with vertices $d_1,d_2,d_3\in\S^2$ we have to decide whether $x_\ell\in \Delta$ and if this is not the case, we have to test whether 
$x_\ell \in \big\{ \lambda_1 d_i + \lambda_2 d_j + \lambda_3 \left( d_i\times d_j \right) \, \big| \, \lambda_1,\lambda_2\ge 0, \lambda_3 \in \R \big\}$ for $(i,j)\in\left\{ (1,2), (2,3), (3,1) \right\}$ and for all $\ell=1,\dots,n$. For the second case, we have to compute the cross product of the vertices anyway, so here it is more efficient to take the criterion of Lemma~\ref{lemma:1}~(b) for testing. Then we only have to compute three scalar products instead of solving a linear equation. To evaluate the last three tests, we calculate a $QR$ decomposition of the matrix $\begin{pmatrix}
d_i & d_j & d_i\times d_j
\end{pmatrix}$ where $Q$ is orthogonal and $R$ is an upper triangular matrix. Using this, the linear equation can be solved quickly via back-solving $Rx = Q^\top x_\ell$ for $x$. De facto, we calculate there also only three scalar products for the right hand side; the additional effort consists in computing the $QR$ decomposition. But since the decomposition remains the same for all $\ell$, this seems to be worthwhile especially for large $n$.

Further technicalities like considering numerical inaccuracies have been taken into account in the implementation used in Section~\ref{sec:5}.

\section{Numerical test on simulated and real data}
\label{sec:5}
Now, we will compare the performance of our $\S$BB algorithm for the case $p = 2$ given four different types of simulated datasets where the points are drawn as indepent and identically distributed samples from differrent distributions on the sphere. In particular, we consider a uniform distribution on a half sphere as well as on the whole sphere, a sample of points forming a tetrahedron rotated around a random angle, and two random points lying diametrically opposed. For points continuously distributed on the sphere one obtains an almost surely unique Fréchet mean, cf.~\cite[Corollary 2.3]{Arnaudon14}. In the first case, this is even the unique local minimizer of the Fréchet function. One therefore would expect a fast convergence behaviour of $\S$BB. For uniformly distributed data on the sphere, one obtains a statement about the average performance for determining a single Fréchet mean. In the other two cases, there are non-unique Fréchet means -- for the tetrahedron there are four, for the diametrically opposed points even infinitely many minimisers, namely all points on the corresponding equator.

Recall that for drawing samples from a uniform distribution on the sphere it suffices to generate three-dimensional standard normally distributed vectors and normalise them afterwards. 

For our simulation we used our own implementation in the programming language \texttt{R}, cf. \cite{R}, a computer with operating system Windows~7, Intel~Core i5-5200U 2{.}2~GHz CPU and 4~GB RAM. The results are listed in Table~\ref{tab:results} and are based on 100 repetitions each. Besides of the distribution the sample was drawn from, we also state the sampling size $n$, the theoretical number $q$ of Fréchet means, the runtime, the number of iterations and the surface measure $\nu(A)$ of the $(\eps,\delta)$-approximation $A$ for $\eps = \delta = 10^{-1}$ in relation to the whole surface of the sphere. All results are listed as means $\pm$ empirical standard deviation.

\begin{table}[H]
	\centering
	\begin{tabular}[h]{l|c|c|c|c|c}
		distribution & $n$ & $q$  & time (in $s$ $\pm$ s.d.) & iterations ($\pm$ s.d.) & $\nu(A)$ ($\pm$ s.d.)\\
		\hline uniform on & \phantom{0}10 & 1 & $\phantom{000.}3{.}0\pm \phantom{00} 1{.}6$ & $\phantom{00.}468\pm\phantom{0.}145$ & $\phantom{0}1{.}3 \% \pm 0{.}1 \%$ \\
		half sphere & 100 & 1 & $\phantom{00.}13{.}8\pm \phantom{00}2{.}2$ & $\phantom{00.}539\pm\phantom{00.}44$ & $\phantom{0}1{.}3 \% \pm 0{.}1 \%$ \\
		\hline uniform on & \phantom{0}10 & 1 & $\phantom{00.}23{.}1\pm \phantom{0}22{.}5$ & $\phantom{0}1{,}356 \pm \phantom{0.}696$ & $\phantom{0}1{.}9\% \pm 0{.}5 \%$ \\
		sphere & 100 & 1 & $\phantom{0.}309{.}2\pm 239{.}9$ & $\phantom{0}5{,}142\pm 2{,}381$ & $\phantom{0}4{.}7 \%\pm 1{.}7\% $ \\
		\hline tetrahedron & \phantom{00}4 & 4 & $1{,}202{.}0\pm \phantom{0}96{.}2$ & $11{,}791\pm \phantom{0.}106$ & $11{.}6 \%\pm 0{.}1\% $  \\
		\hline diametrical & \phantom{00}2 &$\infty$ & $2{,}438{.}6\pm 243{.}0$ &  $13{,}927\pm \phantom{0.}646$ & $19{.}1 \%\pm 0{.}3\% $\\
	\end{tabular}
	\caption{Performance of the $\S$BB algorithm for computing an $(\eps,\delta)$-approximation $A$ of the set of Fréchet means, $\eps=\delta = 10^{-1}$}
	\label{tab:results}
\end{table}

Obviously, the computation of Fréchet means with the $\S$BB algorithm took more time for samples coming from a uniform distribution on the sphere than for those originating from a uniform distribution on a half sphere. This can be explained by the great number of local minimisers of the Fréchet function in the first case which are eliminated from the list $\LO$ only after a long runtime. Nevertheless, in both cases $\S$BB computes in pretty short time an approximation of the minimiser of moderate surface measure although we chose low accuracies of only $\eps=\delta=0.1$.

Note that for our second type of distribution the empirical standard deviation is relatively high. Since the location of the input data has a high influence on the properties of the Fréchet function, local minimisers can differ in their function value only slightly from the global minimum  whereas in some cases this might not happen. In the first case mentioned, the local minimisers remain in the list $\LO$ for a much longer time which results in a longer running time for $\S$BB.

For non-unique Fréchet means, the $\S$BB algorithm takes a much higher amount of iterations for calculating an approximation of these. This is especially the case for our sample of diametrically opposed points where we have to approximate a whole great circle instead of a finite number of points. This is reflected in a much higher effort the one required for approximating a unique Fréchet mean.

We tested our algorithm also for computing the Fréchet means of a real data set called \texttt{DataB5FisherSpherical} which can be obtained from the \texttt{R} package \texttt{CircNNTSR}, see also \cite[Example B5]{Fisher87}. It contains $n=52$ measurements of magnetic remanence from specimen of sedimentary rock in Queensland. The aim is to determine the direction of the North Magnetic Pole at the time of the rock's formation. Because the sample is originated from the same region the mean direction of the magnetic field would give an approximation for this.

For illustration purposes the data is plotted in Figure~\ref{fig:1}. The depicted points are to be understood as directions of the magnetic field, more precisely as the North Magnetic Poles, of the given specimen. Points shown greyed out are on the opposite site of the spherical surface. The output of the $\S$BB algorithm is shown in Figure~\ref{fig:1}. The $(\eps,\delta)$-approximation covers a small area around the Geographic North Pole. Together with the fact that this is connected, one can derive that there is possibly a unique minimiser of the Fréchet-1-function. These data therefore suggest that the North Magnetic Pole at the time of the rock's formation lay in that area.

\begin{figure}
    \centering  
    \includegraphics[width = \textwidth]{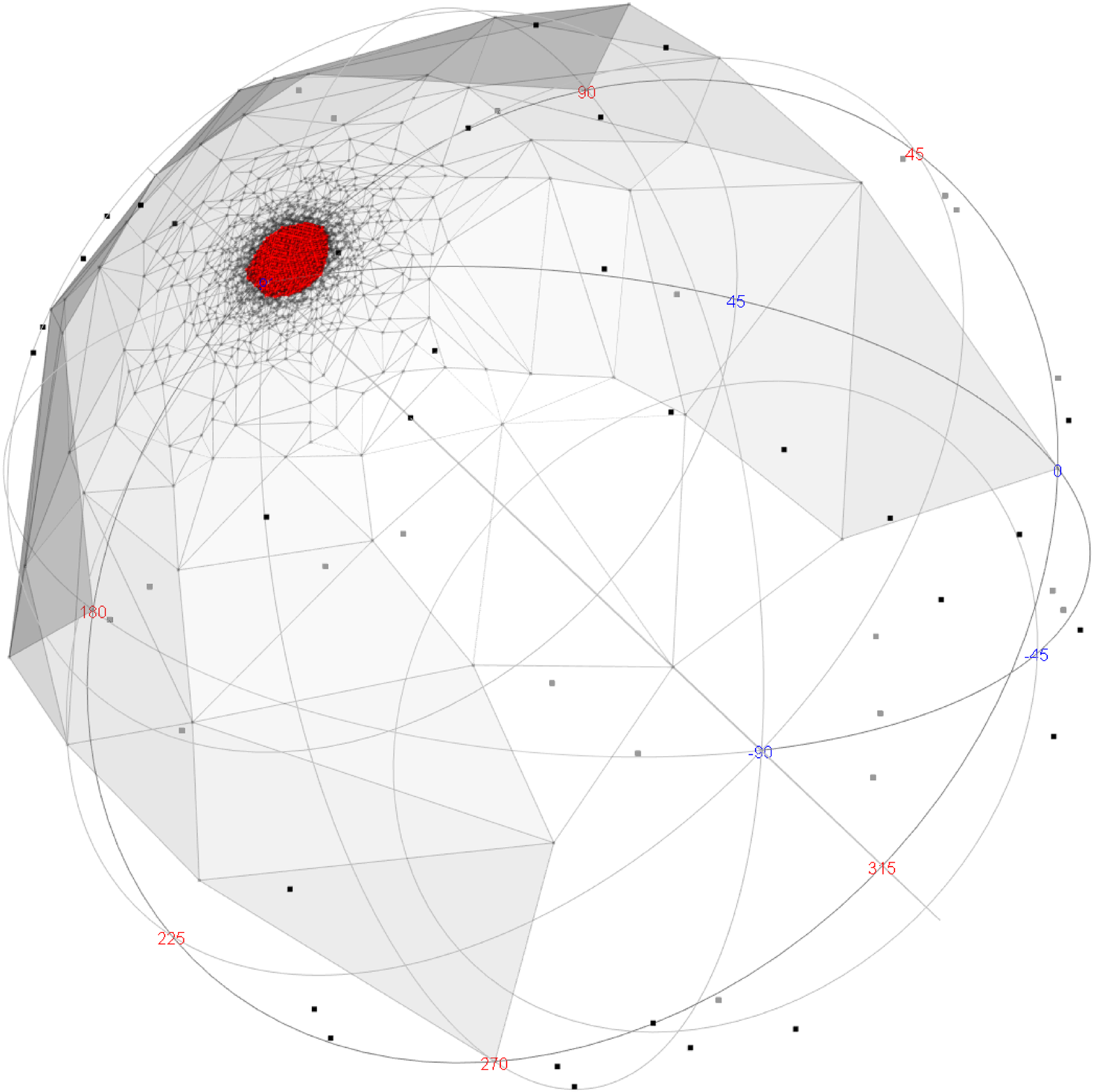}
	\caption{$(\eps,\delta)$-approximation (red) for the Fréchet-1-mean of \texttt{DataB5FisherSpherical} computed by the $\S BB$ algorithm, $(\eps,\delta)=(10^{-2},10^{-1})$, triangles obtained from branching are visualised as Euclidean triangles in grey, the data points given are shown as ($\blacksquare$); longitudes and latitudes are sketched in red and blue, respectively.}
	\label{fig:1}  
\end{figure}

\section{Conclusions and discussion}
\label{sec:6}
We have derived a Branch and Bound method which determines all Fréchet-$p$-means on the sphere for a given finite set of points. The main advantage of $\S$BB is that we do not have to make any assumptions on the given points or the objective function since our lower bound is derived only via geometric considerations. Because of the very general structure of Algorithm~\ref{alg:algorithm}, it will be applicable to other metric spaces too, as long as lower bounds and branching rules are known. In particular, extensions to higher dimensional spheres, real projective spaces, or simplices appears straightforward. We noticed that the running time of $\S$BB depends severely on how the points are distributed on the sphere. This is due to the fact that the location of the points influences the properties of the Fréchet function like number of local and global minimizers. Taking this information into account could improve the performance of our algorithm but requires appropriate criteria or statements for these properties, respectively.

\bibliography{quellen}   

\end{document}